\documentclass[12pt,a4paper]{article}

\usepackage{cite}
\usepackage{amsmath}
\usepackage{amsthm}
\usepackage{amssymb}
\usepackage[utf8]{inputenc}
\usepackage{tikz}
\usepackage{subcaption}

\newtheorem{theorem}{Theorem}
\newtheorem{lemma}[theorem]{Lemma}

\def\defn#1{{\bf #1}}
\def\multi#1{{\begin{tabular}{@{}l@{}}#1\end{tabular}}}

\def\Real{{\mathbb R}}

\def\Norm#1{\left\|#1\right\|}
\def\innerprod(#1,#2){{\left<#1\,,\,#2\right>}}
\def\Set#1{{\left\{#1\right\}}}
\def\qquadtext#1{\qquad\textup{#1}\qquad}
\def\qquadand{\qquadtext{and}}
\def\quadtext#1{\quad\textup{#1}\quad}
\def\quadand{\quadtext{and}}
\def\pfrac#1#2{\frac{\partial #1}{\partial #2}}

\def\pqfrac#1#2#3{\frac{\partial^2 #1}{\partial {#2}\partial{#3}}}
\def\pqqfrac#1#2#3#4{\frac{\partial^3 #1}
  {\partial {#2}\partial{#3}\partial{#4}}}
\def\dfrac#1#2{\frac{d #1}{d #2}}

\usepackage{fancyhdr}
\fancypagestyle{firststyle}
{
   \fancyhf{}
   
\fancyfoot[C]{1}
\fancyfoot[R]{\footnotesize\bf File: \jobname.tex}
}
\newenvironment{jgitemize}{\begin{list}{$\bullet$}
{
\setlength{\topsep}{0.3em}
\setlength{\parsep}{0em}
\setlength{\itemsep}{0.3em}
\setlength{\itemindent}{2em}
\setlength{\leftmargin}{1em}
}}{\end{list}}

\def\mono{{\textup{M}}}
\def\dip{{\textup{D}}}
\def\edp{{\textup{ED}}}
\def\qp{{\textup{Q}}}
\def\eqp{{\textup{EQ}}}

\def\J{{\cal J}}
\def\Jmp{{\cal J}_{\mono}}
\def\Jdp{{\cal J}_{\dip}}
\def\Jed{{\cal J}_{\edp}}
\def\Jqp{{\cal J}_{\qp}}
\def\Jeq{{\cal J}_{\eqp}}

\def\JCmp{{J}_{\mono}}
\def\JCdp{{J}_{\dip}}

\def\JCqp{{J}_{\qp}}

\def\Cdot{\dot{C}}
\def\DomC{{{\cal I}}}
\def\DomChat{{\hat{\cal I}}}
\def\VE{{\boldsymbol E}}
\def\VD{{\boldsymbol D}}
\def\VB{{\boldsymbol B}}
\def\VH{{\boldsymbol H}}

\def\gamhat{\hat{\gamma}}
\def\xhat{\hat{x}}
\def\tauhat{\hat\tau}
\def\phihat{\hat\phi}

\def\Jcurr{{\zeta}}

\def\poleset#1#2{\left\{\text{
{\begin{tabular}{@{\!}c@{}}#1\\#2\end{tabular}}}\right\}}

\def\Vx{{\boldsymbol x}}
\def\VJ{{\boldsymbol J}}
\def\Vp{{\boldsymbol p}}
\def\VT{{\boldsymbol T}}
\def\VA{{\boldsymbol A}}
\def\Vzero{{\boldsymbol 0}}
\def\mdp{{\textup{MD}}}
\def\mqp{{\textup{MQ}}}
\def\tqp{{\textup{Tor}}}

\title{The correct and unusual coordinate transformation rules for electromagnetic
  quadrupoles} 
\author{Jonathan  Gratus\thanks{j.gratus@lancaster.ac.uk}\ \ and Thomas Banaszek
\\
Physics Department, Lancaster University, Lancaster LA1
  4YB, UK.
\\
Cockcroft Institute, Keckwick Lane, Daresbury, WA4 4AD, UK.}

\begin{document}

\maketitle

\begin{abstract}
Despite being studied for over a century, the use of quadrupoles have been
limited to Cartesian coordinates in flat spacetime due to the incorrect
transformation rules used to define them.  Here the correct transformation rules
are derived, which are particularly unusual as they involve second derivatives
of the coordinate transformation and an integral. Transformations involving
integrals have not been seen before. This is significantly different from the
familiar transformation rules for a dipole, where the components transform as
tensors.  It enables quadrupoles to be correctly defined in general relativity
and to prescribe the equations of motion for a quadrupole in a coordinate system
adapted to its motion and then transform them to the laboratory coordinates.  An
example is given of another unusual feature: a quadrupole which is free of
dipole terms in polar coordinates has dipole terms in Cartesian coordinates.  It
is shown that dipoles, electric dipoles, quadrupoles and electric quadrupoles
can be defined without reference to a metric and in a coordinates free
manner. This is particularly useful given their complicated coordinate
transformation.
\end{abstract}

\def\figquadflowrectified{a}
\def\figquadflowLab{b}


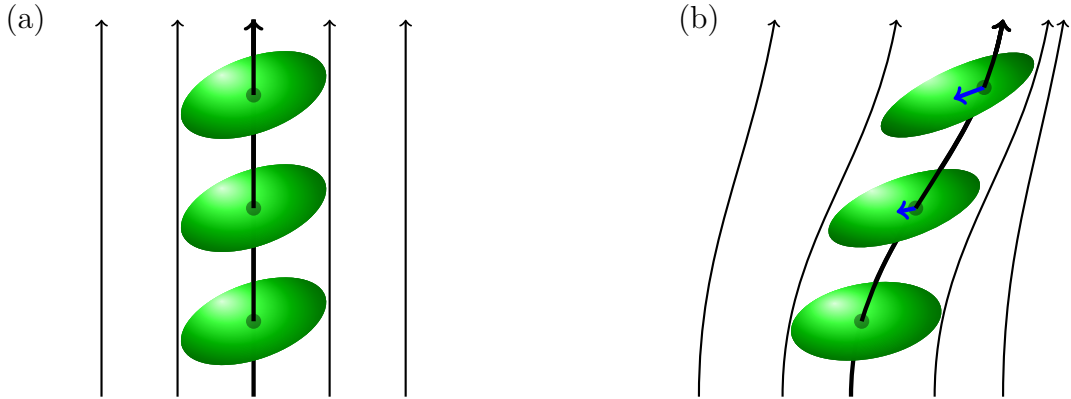
\begin{figure}[t!]
\centering
\begin{subfigure}[t]{0.45\textwidth}
\centering
\begin{tikzpicture}
\draw (-1,5) node {(\figquadflowrectified)} ;
\draw[thick,->] (0,0) -- +(0,5) ;
\draw[thick,->] (1,0) -- +(0,5) ;
\draw[ultra thick,->] (2,0) -- +(0,5) ;
\draw[thick,->] (3,0) -- +(0,5) ;
\draw[thick,->] (4,0) -- +(0,5) ; 
\shade[ball color=green,opacity=0.9] [shift={(2,1)}] [rotate=20,yscale=0.5] (0,0) circle (1) ;
\fill[opacity=0.4] [shift={(2,1)}] (0,0) circle (.1) ;
\draw[ultra thick] (2,1) -- (2,5) ;
\shade[ball color=green,opacity=0.9] [shift={(2,2.5)}] [rotate=20,yscale=0.5] (0,0) circle (1) ;
\fill[opacity=0.4] [shift={(2,2.5)}] (0,0) circle (.1) ;
\draw[ultra thick] (2,2.5) -- (2,5) ;
\shade[ball color=green,opacity=0.9] [shift={(2,4)}] [rotate=20,yscale=0.5] (0,0) circle (1) ;
\fill[opacity=0.4] [shift={(2,4)}] (0,0) circle (.1) ;
\draw[ultra thick] (2,4) -- (2,5) ;
\end{tikzpicture}
\end{subfigure}
\hspace{2em}
\begin{subfigure}[t]{0.45\textwidth}
\centering
\begin{tikzpicture}
\draw (0,5) node {(\figquadflowLab)} ;
\draw[thick,->] (0,0) to [out=90,in=-100] +(1,5) ;
\draw[thick,->] (1.1,0) to [out=90,in=-100] +(1.5,5) ;
\draw[ultra thick,->] (2,0) to [out=90,in=-100] +(2,5)  ;
\draw[thick,->] (3.1,0) to [out=90,in=-100] +(1.5,5)  ;
\draw[thick,->] (4,0) to [out=90,in=-100] +(0.8,5)  ;
\shade[ball color=green,opacity=0.9] [shift={(2.2,1)}] [rotate=10,yscale=0.5] (0,0) circle (1) ;
\fill[opacity=0.4] [shift={(2.14,1)}] (0,0) circle (.1) ;
\draw[ultra thick,->] (2.14,1) to [out=73,in=-100] (4,5)  ;
\shade[ball color=green,opacity=0.9] [shift={(2.7,2.5)}]
      [rotate=20,yscale=0.4,xscale=1.05] (0,0) circle (1) ;
\fill[opacity=0.4] [shift={(2.85,2.5)}] (0,0) circle (.1) ;
\draw [ultra thick,blue,->] [shift={(2.85,2.5)}] (0,0) -- +(-0.25,-0.05) ;
\draw[ultra thick,->] (2.85,2.5) to [out=58,in=-100] (4,5)  ;
\shade[ball color=green,opacity=0.9] [shift={(3.4,4)}] 
      [rotate=25,yscale=0.35,xscale=1.1] (0,0) circle (1) ;
\fill[opacity=0.4] [shift={(3.75,4.1)}] (0,0) circle (.1) ;
\draw [ultra thick,blue,->] [shift={(3.75,4.1)}] (0,0) -- +(-0.4,-0.15) ;
\draw[ultra thick,->] (3.75,4.1) to [out=70,in=-100] (4,5)  ;
\end{tikzpicture}
\end{subfigure}
\caption{Color: Flow of a quadrupole in coordinates adapted to the
  flow (a) and laboratory
  coordinates (b). Observe
  the appearance of a dipole (arrows) in the laboratory coordinate.
Here quadrupoles are represented by
  ellipsoids, and dipoles by arrows. Clearly the equations of
  motion are far simpler in the adapted coordinate system.}
\label{fig_quadflow}
\end{figure}

\begin{table*}[t!]
\begin{equation*}
\begin{array}{|l|c|c|c|}
\hline
\text{Multipole} & \text{charge distribution} & 
\text{current distribution} 
& \text{\footnotesize\multi{Number\\ of\\ components}}
\\
\hline\rule[-.4em]{0em}{1.4em}
\text{\multi{Electric\\Monopole:}}
& \rho_\mono=q\,\delta(\Vx)
& \VJ_\mono=\Vzero 
& 1
\\
\hline\hline\rule[-.6em]{0em}{1.6em}
\text{\multi{Electric\\dipole:}}\qquad
& \rho_\edp=\Vp_\edp\cdot\nabla\delta(\Vx)
& \VJ_\edp=\Vzero 
& 3
\\
\hline\rule[-.6em]{0em}{1.6em}
\text{\multi{Magnetic \\dipole:}}
& \rho_\mdp=0
& \VJ_\mdp=\Vp_\mdp\times\nabla\delta(\Vx)
& 3
\\
\hline\hline\rule[-1.2em]{0em}{3.em}
\text{\multi{Electric\\quadrupole:}}
& \rho_\eqp=\gamma^{0\mu\nu}\,
\displaystyle\pqfrac{\delta}{x^\mu}{x^\nu}
& \VJ_\eqp=\Vzero 
& 6
\\
\hline\rule[-1.2em]{0em}{3.em}
\text{\multi{Magnetic \\quadrupole:}}
& \rho_\mqp=0
& J^\mu_\mqp=\gamma^{\mu\nu\sigma}\,
\displaystyle\pqfrac{\delta}{x^\nu}{x^\sigma}
& 8
\\\hline
\end{array}
\end{equation*}
\begin{equation*}
\begin{array}{|l|c|c|c|}
\hline
\text{Multipole} & \text{\multi{Electric\\Potential}}
& \text{\multi{Magnetic\\Potential}}
& \text{\multi{Falloff of potentials\\as $r\!\to\!\infty$}}
\\
\hline\rule[-.4em]{0em}{1.4em}
\text{\multi{Electric\\Monopole:}}
& \phi_\mono=\frac{q}{4\pi\epsilon_0 r}
& \VA_\mono=\Vzero
& \sim r^{-1}
\\
\hline\hline\rule[-.6em]{0em}{1.6em}
\text{\multi{Electric\\dipole:}}\qquad
& \phi_\edp=\Vp_\edp\cdot\nabla\phi_\mono
& \VA_\edp=\Vzero
& \sim r^{-2}
\\
\hline\rule[-.6em]{0em}{1.6em}
\text{\multi{Magnetic \\dipole:}}
& \phi_\mdp=0
& \VA_\mdp=\Vp_\mdp\times\nabla\phi_\mono
& \sim r^{-2}
\\
\hline\hline\rule[-1.2em]{0em}{3.em}
\text{\multi{Electric\\quadrupole:}}
& \phi_\eqp=\gamma^{0\mu\nu}\,
\displaystyle\pqfrac{\phi_\mono}{x^\mu}{x^\nu}
& \VA_\eqp=\Vzero
& \sim r^{-3}
\\
\hline\rule[-1.2em]{0em}{3.em}
\text{\multi{Magnetic \\quadrupole:}}
& \phi_\eqp=0
& \VA^\mu_\eqp=\gamma^{\mu\nu\sigma}\,
\displaystyle\pqfrac{\phi_\mono}{x^\nu}{x^\sigma}
& \sim r^{-3}
\\\hline
\end{array}
\end{equation*}
\caption{Sources of static electric and magnetic dipoles and
  quadrupoles at the origin and their corresponding potential
  fields. Here $r=\Norm{\Vx}$ and the components $\gamma^{abc}$
  satisfy the symmetry condition (\ref{Intro_QP_gamma_sym}). Here
  $a=0,1,2,3$ and $\mu=1,2,3$. Even in the
static case we can see that there are three electric dipoles, three
magnetic dipoles, six electric quadrupoles and eight magnetic
quadrupoles. }
\label{tab_Intro_Static_Poles}
\end{table*}

\section{Introduction}
\label{ch_Intro}

Multipole expansions are used extensively as an approximation of
extended particles where the mass or charge is considered to be
concentrated at one point. Multipoles have been used in classical
electrodynamics
\cite{jackson1975electrodynamics,ellis1963fields,ellis1966electromagnetic},
quantum mechanics \cite{raab2005multipole}, as a model for
polarisation and magnetisation
\cite{jackson1975electrodynamics, graham1992multipole,
  de2006surprises, raab2005transformed, DEGROOT19651713, Ahmedov2002},
in determining the stucture of molecules in chemistry 
\cite{mitxelena2016molecular,piris2006new,maroulis2018electric,romero2017calculations}
and recently in the idea of meta-atoms
\cite{papasimakis2016electromagnetic}. A distribution of charge can be
approximated by a point charge together with a sum of moments \cite{van1991singular}. The
first correction is called the dipole, the second order correction a
quadrupole and so on.  Three of the magnetic quadrupoles are identified as
toroidal moments
\cite{dubovik1990toroid,heras1998electric,papasimakis2016electromagnetic}
and controversially called toroidal dipoles.
The sources and potentials for electric
and magnetic multipoles at rest (in flat space) are given in table
\ref{tab_Intro_Static_Poles}. (See also \cite{skinner1989introduction})

What are the correct equations of motion for a quadrupole and higher
order moments?  Although the equations of motion for the force and
torque on a dipole are well established
\cite{jackson1975electrodynamics}, the equivalents for quadrupoles is
much less clear. One method is to consider that quadrupoles evolve due
to a flow, as depicted in figure \ref{fig_quadflow}. The easiest
method for analysing this quadrupole motion and evolution is to choose
coordinate systems adapted to the flow, i.e. rectify the flow. In this
coordinate system the quadrupole simply progresses unchanged as in
figure \ref{fig_quadflow}(\figquadflowrectified).  One then needs to
transform this equation into the laboratory coordinate system. For
example to construct the equivalent of the Liénard-Wiechart fields,
\cite{ellis1966electromagnetic}. See figure
\ref{fig_quadflow}(\figquadflowLab).  However to do this one needs the
correct coordinate transformation rules. {\em The primary goal of this
  article is to establish the correct coordinate
  transformations}. This is important not only for transforming
between an adapted and laboratory coordinate systems, but also between
the spherical polars and Cartesian coordinates in flat space and also
between the arbitrary coordinate systems in general relativity.

The components of a dipole, transform in the familiar
way as tensors. That is, using the Jacobian matrix, the entries of
which are the partial
derivatives of the coordinate transformation. One may naturally assume
that the components of a quadrupole transform in a similar
manner. Indeed this is the case if one limits oneself to Lorentz boosts
and rotations in flat Minkowski space. Furthermore in such cases a
pure quadrupole, i.e. one that contains no dipole terms, would remain
a pure quadrupole in all coordinate systems. However with quadrupoles
we have the following unusual properties:
\begin{jgitemize}
\item[{\bf I}:]
The coordinate transformations of quadrupoles require {\em the
  derivatives of the Jacobian matrix and an integral}.  Although
second derivatives of the coordinate functions are familiar for
Christoffel symbols and jet bundles, those involving integrals have,
as far as the authors are aware, never been considered before.
\item[{\bf II}:]
There is no such thing as pure quadrupole. The coordinate 
transformation of a 
quadrupole moment will, in general, produce a dipole moment.
\end{jgitemize}
Since the correct coordinate transformations for quadrupoles have been
unknown up to now, the use of multipole expansions has been limited to
Cartesian coordinates in flat space. This work, therefore will greatly
expand the role of quadrupoles, so that they can be used to model
extended charges in arbitrary coordinate systems and in arbitrary
spacetimes. The tools developed in this article will enable
researches to extend the results to higher order multipoles.

With regards to point {\bf II} above, we give an example of a
quadrupole which in polar coordinates has no dipole terms, whereas in
Cartesian coordinates does have dipole terms. As suggested in figure
\ref{fig_quadflow}, the quadrupole which flows unchanged in the the
adapted coordinates gains a dipole in the laboratory frame. Although
such a fluid flow is uncommon in electromagnetism, it is the natural
Vlasov description for the dynamics of a distribution of charge in
seven dimensional phase-space-time. The extension of the coordinate
transformations to seven dimensions can be easily handled using the
coordinate free approach detailed in this article.  In the adapted
coordinate system, one can also add additional forces, modelling
internal collisions or self forces.

As well as the equation of motion for multipoles one may ask what are
the electromagnetic fields for generated by them.
The electromagnetic fields due to an arbitrary moving dipole, in flat
Minkowski spacetime, were
first calculated by Ellis \cite{ellis1963fields}, and have been
re-derived and re-expressed many times since
\cite{ward1965electromagnetic, monaghan1968heaviside,
  heras1994radiation, power2001electromagnetic}. The fields due to an
arbitrary moving quadrupole or higher multipole were also derived by
Ellis \cite{ellis1966electromagnetic}. He derived the
electromagnetic fields due to an arbitrary moving multipole in
Minkowski spacetime by differentiating the Liénard-Wiechart
fields. As stated these results require that the quadrupole is
expressed in Cartesian coordinates.
However quadrupoles have been much less examined in the
literature. 


Is it is common to separate out dipoles into three electric and three
magnetic dipoles. There is a choice however as to how to do this, one
can separate them out with respect to the rest frame of the dipole
\cite{ellis1963fields} or with respect to a laboratory frame
\cite{monaghan1968heaviside}. Since the components of a dipole
transform covariantly these are easy to perform.  One can also
separate out the quadrupoles with respect to the rest frame or the
laboratory frame. The complicated transformation rules however
mean that these will mix with the dipoles under change of
coordinates. Thus the question of what is an electric or magnetic
quadrupole in an arbitrary spacetime is more subtle. We 
show that, with respect to the particle rest frame the electric
quadrupole is well defined.  As a result, multipoles form a natural hierarchy
\begin{equation}
\begin{aligned}
&\poleset{Electric dipoles}{(Dim=3)}
\text{\LARGE{$\subset$}}
\poleset{All dipoles}{(Dim=6)}
\text{\LARGE{$\subset$}}
\poleset{Electric quadrupoles}{(Dim=12)}
\text{\LARGE{$\subset$}}
\poleset{All quadrupoles}{(Dim=20)}
\end{aligned}
\label{Intro_sets}
\end{equation}

\vspace{1em}

There is considerable interest in which aspect of electrodynamics can
be defined without the use of a metric and hence without gravity
\cite{hehl2012foundations}. If one relaxes the requirement that $\tau$
be proper time, then we see that monopoles,
dipoles and quadrupoles do not require a metric for their
definitions. Indeed even electric dipoles and electric quadrupoles can
be defined without reference to a metric. By contrast the magnetic
multipoles require either a metric or a preferred coordinate system to
define them. The advantage of such definitions are many fold:
\begin{jgitemize}
\item
In general
relativity, the stress-energy-momentum tensor can be derived by a
variation of the metric in the Lagrangian. Knowing that multipoles are
metric-free objects makes the variation much simpler. 
\item
The definitions given mean that the concept of multipoles and electric
multipoles can be generalised not only to higher dimensional
spacetimes but also to manifolds such as phase space or contact
manifolds where there is no preferred metric. In particular one can
talk about multipole expansions of plasmas and beams of particles,
where one takes moments of a probability distribution function in
phase space.
\end{jgitemize}

When dealing with physical objects in arbitrary spacetimes, one has the
choice either to define them with respect to a coordinate system and
then give the coordinate transformations or to define them in a
coordinate-free manner. Thus it is perfectly acceptable to define
quadrupoles using coordinates.
However such complicated transformations rules strongly promotes the
coordinate-free definition of quadrupoles. In this article we give
such a coordinate-free definition.

\vspace{1em}

This article is arranged as follows:

In section \ref{ch_IQuad} we present quadrupoles in the standard
notation using coordinates and an integral over the worldline.

In section \ref{ch_Coord} we derive the general coordinate
transformation for quadrupoles. We also show which quadrupoles are in
fact dipoles and which quadrupoles are electric dipoles.

In section \ref{ch_CoFree} we demonstrate a more abstract property of
dipoles and quadrupoles, that is that they can be defined without
reference either to a coordinate system or to a metric.

Finally in section \ref{ch_Conc} we conclude with some discussion and
suggestion of future research. In the appendix we prove some of the
the more technical statement from section \ref{ch_CoFree}.

\section{The standard representation of Quadrupoles}
\label{ch_IQuad}

In this article the Greek indices
$\mu,\nu,\sigma=1,2,3$ and the Latin indices $a,b,c=0,1,2,3$. We use
the summation convention with implicit summation over pairs of
matching high and low indices, unless otherwise stated.

The static electric and magnetic dipoles, table \ref{tab_Intro_Static_Poles},
can be combined into a single two component antisymmetric tensor
$\gamma^{ab}$, 
\begin{align}
\gamma^{ab}+\gamma^{ba}=0
\label{Intro_DP_gamma_sym}
\end{align}
where
\begin{align}
\gamma^{0\mu}=\Vp_\edp^\mu
\qquadand
\gamma^{\mu\nu}=\epsilon^{\mu\nu\sigma}\,(\Vp_\mdp)_\sigma
\label{Intro_gamma_ab_Pdp}
\end{align}

The quadrupoles components in table \ref{tab_Intro_Static_Poles} can
be combined into a single three component object $\gamma^{abc}$.
Due to conservation of charge, 
these satisfy the symmetry conditions:
\begin{align}
\gamma^{abc}=\gamma^{acb}
\qquadand
\gamma^{abc}+\gamma^{bca}+\gamma^{cab}
=0
\label{Intro_QP_gamma_sym}
\end{align}
The symmetry conditions (\ref{Intro_QP_gamma_sym}) give eight quadrupole
components. These may be written
$\gamma^{\mu\nu\nu}$
for $\mu\ne\nu$
(no sum) which give six components
and the pair
($\gamma^{123}$, $\gamma^{231}$).
So that from (\ref{Intro_QP_gamma_sym})
$\gamma^{\nu\mu\nu}=-\tfrac12\gamma^{\mu\nu\nu}$ and
$\gamma^{312}=\gamma^{123}+\gamma^{231}$.

Toroidal ``dipoles'' are given by
\begin{align}
\gamma^{\mu\nu\sigma}_\tqp 
= 
T^{\nu}\,\delta^{\mu\sigma} + T^{\sigma}\,\delta^{\mu\nu}
-2 T^{\sigma}\,\delta^{\nu\sigma}
\label{Intro_Tor_quads}
\end{align}
Substituting (\ref{Intro_Tor_quads}) into $\VJ_\mqp$ (table
\ref{tab_Intro_Static_Poles}) gives
$\VJ_\tqp=\nabla\times\nabla\times(\VT\delta)$. 
In our classification these are considered
quadrupole terms.
In \cite{dubovik1990toroid,heras1998electric,papasimakis2016electromagnetic}
these are
actually referred to as toroidal ``dipoles''. This despite the fact that 1:
they involve the second derivative and 2: their potential fields fall
off as $r^{-3}$. In addition they are not immune from the 
complicated transformations rules investigated in this article.


For moving multipoles the sources given in table
\ref{tab_Intro_Static_Poles} have to be integrated over the worldline.
In Minkowski spacetime, the electromagnetic fields due to a moving electric
charge are known as the Liénard-Wiechart fields. The source for the
Liénard-Wiechart is the 4-current $J^a(x)$ which may be written in
terms of the Dirac $\delta-$function \cite{jackson1975electrodynamics}
\begin{align}
\JCmp^a(x) = 
{q}\int_\DomC \Cdot^a(\tau)\,
\delta\big(x-C(\tau)\big)\, d\tau 
\label{Intro_Jmp_a_delta}
\end{align}
where $C^a(\tau)$ are the components of the worldline of the
particle of charge $q$, $\Cdot^a=\dfrac{C^a}{\tau}$ 
and $x$ is a point in spacetime. 
The parameter
$\tau$ is usually considered to be the proper time of the particle,
although it need not be, and
the interval
$\DomC\subset\Real$ is the range of $\tau$. The source
(\ref{Intro_Jmp_a_delta}) is valid for any spacetime, although in
general the corresponding electromagnetic fields have not been
calculated.

An arbitrary moving dipole in an arbitrary moving spacetime may be written
\begin{align}
\JCdp^a(x)
=
\int_\DomC
\gamma^{ab}(\tau)\, \pfrac{\delta}{x^b}\big(x-C(\tau)\big)\, d\tau
\label{Intro_Jdp_a_delta}
\end{align}
Multiple authors have found the electromagnetic field for a dipole in
Minkowski spacetime in terms of an integral of the retarded Green's
function.

Due to conservation of charge the parameters defining the dipole are
constrained to be antisymmetric (\ref{Intro_DP_gamma_sym}) giving six
components. However once (\ref{Intro_DP_gamma_sym}) is imposed the
components $\gamma^{ab}=\gamma^{ab}(\tau)$ may be arbitrary functions
of $\tau$.
It is easy to show that under a change of basis the
$\gamma^{ab}$ transforms as a tensor, given by
(\ref{Coord_DP_change_coords}) below. This is true both for global
linear transformations, in Minkowski spacetime, $x^a\to \xhat^a =
A^a_b x^b$ and for local coordinate transformations $x^a\to \xhat^a =
\hat{x}^a(x^0,\ldots,x^3)$. 

The generalisation of (\ref{Intro_Jdp_a_delta}) to quadrupoles is less
common. Kaufmann \cite{kaufman1962maxwell} was the first to express the
quadrupole as an expansion to the second derivative of the
$\delta-$function. 
Ellis \cite{ellis1963fields,ellis1966electromagnetic} 
observed that these can be written  
\begin{align}
\JCqp^a
=
\tfrac12\int_\DomC
\gamma^{abc}(\tau)\, \pqfrac{\delta}{x^b}{x^c}\big(x-C(\tau)\big)\, d\tau
\label{Intro_Jqp_a_delta}
\end{align}
He also calculated the electromagnetic fields for arbitrary moving
quadrupoles and higher order moments in Minkowski spacetime, in terms
of the integral of the retarded Green's function. 
The $\gamma^{abc}$ subject to (\ref{Intro_gamma_ab_Pdp}) give twenty
independent quadrupoles components. Again the
$\gamma^{abc}=\gamma^{abc}(\tau)$ may be arbitrary functions of
$\tau$. 

We observe that the quadrupoles
given in (\ref{Intro_Jqp_a_delta}) can also contain dipole terms. This
can be seen in the static case when $t=\tau=x^0$ and so
\begin{align*}
\int_\DomC
&\gamma^{ab0}(\tau)\, \pqfrac{\delta}{x^b}{x^0}\big(x-C(\tau)\big)\, d\tau
=
-\int_\DomC
\dfrac{\gamma^{ab0}}{\tau}\, \pfrac{\delta}{x^b}\big(x-C(\tau)\big)\, d\tau
\end{align*}
Thus the electromagnetic fields due to a general static quadrupole given in
(\ref{Intro_Jqp_a_delta}) will contain both terms that fall off as
distance cubed and terms that fall of as the fourth power of distance.
The twenty quadrupole components split into six dipole components
and fourteen dipolefree-quadrupole components.

Although for global linear transformations in Minkowski spacetime
$x^a\to \xhat^a = A^a_b\, x^b$ the components $\gamma^{abc}(\tau)$
transformation tensorially, this is not true for general coordinate
transformation. As stated, the rules for a general coordinate
transformation requires a second derivative of the coordinate
functions and an integral. These are given in
(\ref{Coord_QP_change_coords}-\ref{Coord_QP_def_P}) below.  Perhaps,
it is less surprising since (\ref{Intro_Jqp_a_delta}) contains a
second derivative and an integral. However it is unexpected when
contrasted with the dipole case, where the components transform as
tensors.  One problem with such a transformation is that it will give
rise to an arbitrary constant of integration. Fortunately this
constant does not effect the resulting quadrupole, as the terms are
subsequently differentiated. Likewise it will not effect the
corresponding electromagnetic fields in flat spacetime.

A quadrupole
which does not appear to contain any dipole terms, will in general, acquire
dipole terms when one performs a change of coordinates. This
contrasts with the monopole term (\ref{Intro_Jmp_a_delta}) which does
not mix with other multipole terms under change of coordinates. 
As a simple example consider a quadrupole at rest given in axial
cylindrical coordinates $(t,r,\theta,z)$ with $\gamma^{abc}=0$ except
$\gamma^{211}=-2\gamma^{121}=-2\gamma^{112}=2\kappa$ where
$\kappa\in\Real$, $\kappa\neq0$ is a constant. Since it contains no
components with a $0$ index one may expect that in Cartesian
coordinates $(t,x,y,z)$ with $x=r\cos\theta$ and $y=r\sin\theta$, it
would not contain any dipole terms. Indeed if one were to assume that
$\gamma^{abc}$ transforms tensorially then this would be the
case. However, with the correct transformation rules, given in
(\ref{Coord_QP_change_coords}-\ref{Coord_QP_def_P}) below, even in this
simple case gives rise to a dipole term. Writing $\gamhat^{ab}$ and
$\gamhat^{abc}$ for the (dipole and quadrupole) components with
respect to Cartesian coordinates, the dipole term is
$\gamhat^{12}=\kappa$. As stated this would give rise to an $r^{-2}$
fall off for the potential. Furthermore, when expressed in terms of
(\ref{Intro_Jqp_a_delta}) the component $\gamhat^{012}= \kappa t +
\kappa_0$ which grows indefinitely and contains an arbitrary constant
$\kappa_0$. Fortunately as stated above these are differentiated away.

As with the monopole and dipole terms, one
can define a quadrupole in an arbitrary spacetime, using
(\ref{Intro_Jqp_a_delta}) with 
$\gamma^{abc}$ subject to (\ref{Intro_QP_gamma_sym}). In this case the
lack of a preferred coordinate system means that one cannot separate
out the dipoles from the quadrupoles. Likewise since we are not in
Minkowski spacetime, one cannot in general, use the fall off of the
corresponding electromagnetic fields to distinguish the terms
either. Thus in this case there is no concept of a dipolefree-quadrupole.

\section{Coordinate Transformations of Quadrupoles}
\label{ch_Coord}

Since the multipoles involve Dirac
$\delta-$functions it is necessary to integrate them with test
functions in order to evaluate them. Recall, these test functions are
smooth and have compact support. That is, they are infinitely
differentiable and are non zero only on a bounded region of
spacetime. We write these test functions as
$(\phi_0,\ldots,\phi_3)$ which are components of a covector. Acting on
$\phi_a$ we have from
(\ref{Intro_Jmp_a_delta}-\ref{Intro_Jqp_a_delta}), 
it is easy to see that the monopole,
dipoles and quadrupoles give
\begin{align}
\int_M \JCmp^a \, \phi_a\, d^4x 
&= 
{q}\int_\DomC \Cdot^a(\tau)\,\phi_a|_{C(\tau)}
\,d\tau 
\label{Coord_Jmp_act_delta}
\\
\int_M \JCdp^a\, \phi_a\, d^4x 
&=
-\int_\DomC
\gamma^{ab}(\tau)\, \pfrac{\phi_a}{x^b}\Big|_{C(\tau)}\, d\tau
\label{Coord_Jdp_act_delta}
\\
\int_M \JCqp^a\, \phi_a\, d^4x 
&=
\tfrac12\int_\DomC
\gamma^{abc}(\tau)\, \pqfrac{\phi_a}{x^b}{x^c}\Big|_{C(\tau)}\, d\tau
\label{Coord_Jqp_act_delta}
\end{align}
where $M$ is spacetime and we have assumed that $\phi_a$ is only non
zero on the coordinate patch $(x^0,\ldots,x^3)$.
Given new coordinates $(\xhat^0,\ldots,\xhat^3)$ then we require that
$\JCmp^a$, $\JCdp^a$ and $\JCqp^a$ all transform as vectors. That is
\begin{align}
\hat{\JCmp^a} = \pfrac{\xhat^a}{x^b}\, \JCmp^b
\,,\quad
\hat{\JCdp^a} = \pfrac{\xhat^a}{x^b}\, \JCdp^b
\quadand
\hat{\JCqp^a} = \pfrac{\xhat^a}{x^b}\, \JCqp^b
\label{Coord_trans_J}
\end{align}
These transformation are automatic in
(\ref{Coord_Jmp_act_delta}-\ref{Coord_Jqp_act_delta}) since we assumed
$\phi_a$ transformed as a covector, i.e.
\begin{align}
\phihat_a = \pfrac{x^a}{\xhat^b} \phi_b
\label{Coord_trans_phi}
\end{align}
We also wish to consider allowing different parametrisation
$\tau\in\DomC$ and $\tauhat\in\DomChat$. For example one may be proper
time and the other lab time.

Thus relating the new and old coordinates then
(\ref{Coord_Jmp_act_delta}-\ref{Coord_Jqp_act_delta}) and
(\ref{Coord_trans_J}),(\ref{Coord_trans_phi})
imply
\begin{align}
&\int_\DomC \Cdot^a(\tau)\,\phi_a|_{C(\tau)}\, d\tau
=
\int_\DomChat \dot{\hat{C}}^a(\tauhat)\,\phihat_a|_{C(\tauhat)}\,
d\tauhat
\label{Coord_trans_Jmp_phi}
\\
&\int_\DomC
\gamma^{ab}(\tau)\, \pfrac{\phi_a}{x^b}\Big|_{C(\tau)}\, d\tau
=
\int_\DomChat
\gamhat^{ab}(\tauhat)\, \pfrac{\phihat_a}{\xhat^b}\Big|_{C(\tauhat)}\, d\tauhat
\label{Coord_trans_Jdp_phi}
\\
&\int_\DomC
\gamma^{abc}(\tau)\, \pqfrac{\phi_a}{x^b}{x^c}\Big|_{C(\tau)}\, d\tau
=
\int_\DomChat
\gamhat^{abc}(\tauhat)\, \pqfrac{\phihat_a}{x^c}{\xhat^b}\Big|_{C(\tauhat)}\, d\tauhat
\label{Coord_trans_Jqp_phi}
\end{align}
The charge $q$ associated with the monopole is invariant under
coordinate transformation, which follows from
(\ref{Coord_trans_Jmp_phi}). As stated in the introduction the coordinate
transformation of dipole components $\gamma^{ab}$ 
is tensorial, i.e.
\begin{align}
\gamhat^{ab}
=
\pfrac{\phihat_c}{\xhat^d}\,
\pfrac{\xhat^d}{x^b}\, \pfrac{\xhat^c}{x^a}\, \dfrac{\tau}{\tauhat} 
\gamma^{ab}
\label{Coord_DP_change_coords}
\end{align}
since
\begin{align*}
\int_\DomC
\pfrac{\phihat_a}{x^b}\, \gamhat^{ab}\, d\tauhat
&=
\int_\DomC
\pfrac{\phi_a}{x^b}\, \gamma^{ab}\, d\tau
=
\int_\DomC
\pfrac{\xhat^d}{x^b}\pfrac{}{\xhat^d}
\Big(\pfrac{\xhat^c}{x^a}\phihat_c\Big)\, \gamma^{ab}\, d\tau
\\&=
\int_\DomC
\Big(
\pfrac{\xhat^d}{x^b}\pfrac{\xhat^c}{x^a}\pfrac{\phihat_c}{\xhat^d}
+
\pqfrac{\xhat^c}{x^a}{x^b}\phihat_c\Big)
\gamma^{ab}\, d\tau
\\&=
\int_\DomC
\pfrac{\xhat^d}{x^b}\pfrac{\xhat^c}{x^a}\pfrac{\phihat_c}{\xhat^d}
\gamma^{ab}\, d\tau
=
\int_\DomC
\pfrac{\phihat_c}{\xhat^d}\,
\pfrac{\xhat^d}{x^b}\, \pfrac{\xhat^c}{x^a}\, \dfrac{\tau}{\tauhat} 
\gamma^{ab}\, d\tauhat
\end{align*}

In contrast to the dipole, the coordinate transformation of the
components $\gamma^{abc}$ of the quadrupole is given by
\begin{align}
\gamhat^{def}
&= 
\dfrac{\tau}{\tauhat}
\Big(
A_{a}^{d}
A_{b}^{e}
A_{c}^{f}
\gamma^{abc}
+
P^{de}\,
\Cdot^f
+
P^{df}\,
\Cdot^e
\Big)
\label{Coord_QP_change_coords}
\end{align}
where
\begin{align}
A_{b}^{a} = \pfrac{\xhat^a}{x^b}\Big|_{C(\tau)}
\,,\quad
A_{bc}^{a} = \pqfrac{\xhat^a}{x^c\,}{ x^b}\Big|_{C(\tau)}
\,,\quad
\Cdot^e=\dfrac{C^e}{\tau}
\label{Coord_QP_def_A}
\end{align}
and
\begin{align}
P^{de}(\tau)
=
\!\!\int^{\tau} \!\!\!
\gamma^{abc}(\tau')\,\Big(\!A^d_c(\tau')\,A^e_{ab}(\tau')
-
A^e_c(\tau')\,A^d_{ab}(\tau')\!\Big)
\,d\tau'
\label{Coord_QP_def_P}
\end{align}
The proof of (\ref{Coord_QP_change_coords}) is as follows:
\begin{align*}
\tfrac12\int_\DomC
\pqfrac{\phi_a}{x^c}{x^b}\, \gamma^{abc}\, d\tau
&=
\tfrac12\int_\DomC
\pfrac{}{x^c}\bigg(\pfrac{}{x^b}\Big(\pfrac{\xhat^d}{x^a}\phihat_d\Big)\bigg)\, 
\gamma^{abc}\, d\tau
\\&=
\tfrac12\int_\DomC
\pfrac{}{x^c}\Big(
\pqfrac{\xhat^d}{x^a}{x^b}\phihat_d
+
\pfrac{\xhat^d}{x^a}\pfrac{\xhat^e}{x^b}\pfrac{\phihat_d}{\xhat^e}
\Big)\, 
\gamma^{abc}\, d\tau
\\&=
\tfrac12\int_\DomC
\Big(\pqqfrac{\xhat^d}{x^c}{x^a}{x^b}\phihat_d
+
\pfrac{\xhat^e}{x^c}\pqfrac{\xhat^d}{x^a}{x^b}
\pfrac{\phihat_d}{\xhat^e}
+
\pqfrac{\xhat^d}{x^a}{x^c}\pfrac{\xhat^e}{x^b}\pfrac{\phihat_d}{\xhat^e}
\\&\qquad\qquad
+
\pfrac{\xhat^d}{x^a}\pqfrac{\xhat^e}{x^b}{x^c}\pfrac{\phihat_d}{\xhat^e}
+
\pfrac{\xhat^d}{x^a}\pfrac{\xhat^e}{x^b}
\pfrac{\xhat^f}{x^c}
\pqfrac{\phihat_d}{\xhat^f}{\xhat^e}
\Big)\, 
\gamma^{abc}\, d\tau
\\&=
\tfrac12\int_\DomC
\Big(
\big(A^e_c\,A^d_{ab}
+
A^e_b\,A^d_{ac}
+
A^d_a\,A^e_{bc}\big)\pfrac{\phihat_d}{\xhat^e}
+
A^{d}_{a}\,A^{e}_{b}\,A^{f}_{c}\,
\pqfrac{\phihat_d}{\xhat^f}{\xhat^e}
\Big)\, 
\gamma^{abc}\, d\tau
\\&=
\tfrac12\int_\DomC
\Big(
\,S^{de}\,\pfrac{\phihat_d}{\xhat^e}
+
A^{d}_{a}\,A^{e}_{b}\,A^{f}_{c}\,\gamma^{abc}\, 
\pqfrac{\phihat_d}{\xhat^f}{\xhat^e}
\Big)\, 
d\tau
\end{align*}
where
\begin{align*}
S^{de}=\big(A^e_c\,A^d_{ab}
+
A^e_b\,A^d_{ac}
+
A^d_a\,A^e_{bc}\big)\gamma^{abc}\, 
\end{align*}
However
\begin{align*}
\big(A^e_c\,A^d_{ab}
+
A^e_b\,A^d_{ac}\big)\gamma^{abc}
&=
A^e_c\,A^d_{ab}\big(\gamma^{abc}+\gamma^{acb}\big)
\\&=
2A^e_c\,A^d_{ab}\gamma^{abc}
\end{align*}
and
\begin{align*}
A^d_a&\,A^e_{bc}\,\gamma^{abc}
=
-A^d_a\,A^e_{bc}(\gamma^{cab}+\gamma^{bca})
\\&=
-A^d_b\,A^e_{ca}\gamma^{abc}-A^d_c\,A^e_{ab}\gamma^{abc}
=
-2A^d_c\,A^e_{ab}\gamma^{abc}
\end{align*}
Hence
\begin{align*}
S^{de}=2\big(A^e_c\,A^d_{ab}
-
A^d_c\,A^e_{ab}\big)\, \gamma^{abc}
\end{align*}
so that $S^{de}+S^{ed}=0$. From (\ref{Coord_QP_def_P})
$\displaystyle S^{de}=-2\dfrac{P^{de}}{\tau}$ giving 
\begin{align*}
\tfrac12\int_\DomC
S^{de}\, \pfrac{\phihat_d}{\xhat^e}
\,d\tau
&=
-\int_\DomC
\dfrac{P^{de}}{\tau}
\pfrac{\phihat_d}{\xhat^e}
\,d\tau
=
\int_\DomC
P^{de}\dfrac{}{\tau}
\pfrac{\phihat_d}{\xhat^e}
\,d\tau
=
\int_\DomC
P^{de}\Cdot^f 
\pqfrac{\phihat_d}{\xhat^f}{\xhat^e}
\,d\tau
\\&=
\tfrac12\int_\DomC
\big(P^{de}\Cdot^f + P^{df}\Cdot^e\big)
\pqfrac{\phihat_d}{\xhat^f}{\xhat^e}
\,d\tau
\end{align*}
Thus
\begin{align*}
\tfrac12\int_\DomC
\pqfrac{\phihat_d}{x^e}{x^f}\, \gamhat^{def}\, d\tauhat
&=
\Jqp[\phi]
=
\tfrac12\int_\DomC
\Big(\!
P^{de}\Cdot^f + P^{df}\Cdot^e
+
A^{d}_{a}\,A^{e}_{b}\,A^{f}_{c}\,\gamma^{abc}\Big)
\pqfrac{\phihat_d}{\xhat^f}{\xhat^e}
\, d\tau
\\&=
\tfrac12\int_\DomC
\dfrac{\tau}{\tauhat}\Big(\!
P^{de}\Cdot^f \!+\! P^{df}\Cdot^e
\!+\!
A^{d}_{a}\,A^{e}_{b}\,A^{f}_{c}\,\gamma^{abc}\Big)
\pqfrac{\phihat_d}{\xhat^f}{\xhat^e}
\, d\tauhat
\end{align*}
which gives (\ref{Coord_QP_change_coords}).\qed

\vspace{1em}

As stated certain quadrupoles are in fact
dipoles. Consider the quadrupole given by (\ref{Coord_Jqp_act_delta})
with $\gamma^{abc}(\tau)$ given by
\begin{align}
\gamma^{abc}= p^{ab}\,\Cdot^c +p^{ac}\,\Cdot^b
\label{Coord_QP_DP_gam}
\end{align}
where $p^{ab}=p^{ab}(\tau)$ and $p^{ab}+p^{ba}=0$. Note that these
satisfy (\ref{Intro_QP_gamma_sym}). The quadrupole
$\JCqp^a$ is in fact a dipole $\JCdp^a$ where
\begin{align}
\gamma^{ab} = \dot{p}^{ab}
\label{Coord_QP_DP_p_int}
\end{align}
where $\dot{p}^{ab}={d p^{ab}}/{d \tau}$. This follows since substituting (\ref{Coord_Jqp_act_delta}) into
\eqref{Coord_QP_DP_gam} gives
\begin{align*}
\tfrac12\int_\DomC&
(p^{ab}\,\Cdot^c +p^{ac}\,\Cdot^b)\,\pqfrac{\phi_a}{x^c}{x^b}\,  d\tau
\\&=
\int_\DomC
p^{ab}\,\Cdot^c \,\pqfrac{\phi_a}{x^c}{x^b}\,  d\tau
=
\int_\DomC
p^{ab}\,\dfrac{}{\tau}\,\Big(\pfrac{\phi_a}{x^b}\Big|_{C(\tau)}\Big)\,  d\tau
\\&=
-\int_\DomC
\dot{p}^{ab}\,\Big(\pfrac{\phi_a}{x^b}\Big|_{C(\tau)}\Big)\,  d\tau
\end{align*}
Hence by comparing with (\ref{Coord_Jdp_act_delta}) we see that $\Jqp$
contains only a dipole term with (\ref{Coord_QP_DP_p_int}).\qed

\vspace{1em}

We can now demonstrate our example of the ``dipolefree-quadrupole'' in
axial cylindrical coordinates outlined in the introduction. Let
$x^a=(t,r,\theta,z)$ be axial cylindrical coordinates and
$\xhat^a=(t,x,y,z)$ be Cartesian coordinates, with transformation
functions $x=r\cos\theta$ and $y=r\sin\theta$. The transformation
rules (\ref{Coord_QP_def_A}) are given by
\begin{align*}
&A^a_b=\delta^a_b
\quadtext{except}
A^1_1 = \pfrac{x}{r} = \cos\theta,\ \,
A^2_1 = \pfrac{y}{r} = \sin\theta,\ \,
\\&
A^1_2 = \pfrac{x}{\theta} = -r\sin\theta,\quad
A^2_2 = \pfrac{y}{\theta} = r\cos\theta
\\
&\text{and}\quad A^a_{bc}=0
\quadtext{except}
A^1_{12} = A^1_{21} = -\sin\theta,\quad
\\&
A^2_{12} = A^2_{21} = \cos\theta,\quad
A^1_{22} = -r\cos\theta,\quad
A^2_{22} = -r\sin\theta
\end{align*}
Thus the integrated in (\ref{Coord_QP_def_P}) corresponding to
$P^{12}$, the only no zero dipole component, is given by
\begin{align*}
\gamma^{abc}&(\tau')\,\Big(A^1_c(\tau')\,A^2_{ab}(\tau')
-
A^2_c(\tau')\,A^1_{ab}(\tau')\Big)\, 
\\&=
\gamma^{112}(A^1_2\,A^2_{11}-A^2_2\,A^2_{11})
+
\gamma^{121}(A^1_1\,A^2_{12}-A^2_1\,A^2_{12})
+
\gamma^{211}(A^1_1\,A^2_{12}-A^2_1\,A^2_{12})
\\&= 
\kappa
\end{align*}
Hence $P^{12}=-P^{21}=\kappa t+\kappa_0$ where $\kappa_0$ is an
arbitrary constant of integration. Hence $\gamhat^{012}=\kappa
t+\kappa_0$. Using (\ref{Coord_QP_DP_gam}) and
(\ref{Coord_QP_DP_p_int}) we see that $P^{21}$ gives rise to a dipole
component with $\gamhat^{12}=\kappa$.

\vspace{1em}

With respect to a coordinate system, the dipole components
$\Set{\gamma^{01},\gamma^{02},\gamma^{03}}$ are electric and
  $\Set{\gamma^{12},\gamma^{13},\gamma^{23}}$ are magnetic.
In this article we only consider splitting the electric and magnetic
components with respect to the instantaneous rest frame of the particle.
In an arbitrary coordinate system, the electric dipole $\Jed$ may be written 
\begin{align}
\gamma^{ab}=w^a \Cdot^b - w^b \Cdot^a
\label{Coord_DP_ED}
\end{align}
where $w^a(\tau)$ transforms as a vector.  Note that replacing
$w^a(\tau)$ with $w^a(\tau)+\xi(\tau)\Cdot^a(\tau)$, for any scalar
$\xi(\tau)$, does not change $\Jed^a$. 

In spacetime there is a preferred rest coordinate system, called the
Fermi coordinates,  about a worldline. 
For quadrupole we say that electric dipoles 
are those which in the Fermi coordinate system have components with a zero,
i.e. $\gamma^{0ab},\gamma^{a0b},\gamma^{ab0}$. From
(\ref{Coord_QP_DP_gam}) we see that these contain all the
dipoles. There are six dipolefree electric quadrupoles. Likewise
the magnetic eight magnetic quadrupole components 
contain only $\gamma^{\mu\nu\rho}$, where
Greek indices run over $\mu,\nu,\rho=1,2,3$,.

It turns out that identifying the electric quadrupoles in an arbitrary
coordinate system is easy. These become
\begin{align}
\gamma^{abc}
=
\Cdot^a\,q^{bc} + \Cdot^a\, q^{cb} - \Cdot^b\, q^{ac} - \Cdot^c\,
q^{ab}
\label{Coord_EQP}
\end{align}
The $q^{ab}(\tau)$ have no restrictions, but the $\gamma^{abc}(\tau)$
is unchanged if we replace
$q^{ab}(\tau)=q^{ab}(\tau)+s^a(\tau)\Cdot^b(\tau)+s^b(\tau)\Cdot^c(\tau)$
for any indexed scalars $s^a(\tau)$.  The $\gamma^{abc}$ given by
(\ref{Coord_EQP}) satisfy the symmetry conditions
(\ref{Intro_QP_gamma_sym}).  This gives the twelve independent
electric quadrupole terms.  If $q^{ab}+q^{ba}=0$ then
(\ref{Coord_EQP}) reduces to (\ref{Coord_QP_DP_gam}), with
$p^{ab}=q^{ab}$.
Equation (\ref{Coord_EQP}) is proved in the appendix after we have
introduced the coordinate-free and metric-free definitions.

\section{Coordinate-free and metric-free definition of multipoles}
\label{ch_CoFree}

As stated in the introduction, in this section we introduce coordinate-free
and metric-free definitions of dipoles, quadrupoles, electric dipoles
and electric quadrupoles. This is because quadrupoles and electric
quadrupoles are much easier to define in a coordinate-free manner.
Let the set of all smooth $p-$form fields on
spacetime $M$ be
written $\Gamma\Lambda^pM$. A test form is a form
$\phi\in\Gamma\Lambda^pM$ with compact support. The set of all test
$p-$forms is written $\Gamma_0\Lambda^pM$.

Since the dipoles and quadrupoles are only non-zero along a worldline
one must use notion of distributions in order to define them. Recall
that the current $3-$form, $\J$, which includes dipoles and
quadrupoles, is the source of Maxwell's equations. In the language of
exterior differential forms, Maxwell's equations become
\begin{align}
dF=0\qquadand dH=\J
\label{Defs_Maxwell}
\end{align}
where $F\in\Gamma\Lambda^2M$ is the electromagnetic 2-form encoding
the electric fields $\VE$ and the magnetic flux density $\VB$, and
where $H\in\Gamma\Lambda^2M$ is the excitation 2-form encoding the
displacement field $\VD$ and the magnetic field intensity $\VH$. Here
$d$ is the exterior derivative. The fields $F$ and $H$ have to be
related by constitutive relations. The constitutive relations for the
vacuum are given by $H=\star F$, where $\star$ is the Hodge dual,
derived from the metric. These lead to the microscopic Maxwell
equations, $d\star F=\J$. Taking the exterior derivative of the second
equation in (\ref{Defs_Maxwell}) leads to the continuity equation
\begin{align}
d\J=0
\label{Defs_dJ}
\end{align}
which in turn leads to conservation of
charge. We say a $\J$ which satisfies (\ref{Defs_dJ}) is \defn{closed}.

Since Maxwell's equations are linear one can consider distributional
currents. Following Schwartz, 
we define a distribution what is does on a test $(4-p)-$form
$\phi\in\Gamma\Lambda^{4-p}M$. A
test $(4-p)-$form has compact support. If
$\alpha\in\Gamma\Lambda^pM$ is a smooth $p-$form, we can construct a
regular distribution $\alpha^D$ via
\begin{align}
\alpha^D[\phi]=\int_M \phi\wedge \alpha
\label{Defs_def_alpha_D}
\end{align}
The definition of the wedge product, Lie derivatives, internal
contraction and exterior derivatives on distributions are defined to
be consistent with (\ref{Defs_def_alpha_D}). Thus for a distribution
$\Psi$ we set
\begin{align}
\begin{gathered}
(\Psi_1+\Psi_2)[\phi]
=
\Psi_1[\phi]+\Psi_2[\phi]
\,,\quad
(\beta\wedge\Psi)[\phi]=\Psi[\phi\wedge\beta]
\,,\quad
(d\Psi)[\phi]=(-1)^{(3-p)}\Psi[d\phi]
\,,
\\
(i_v\Psi)[\phi]=(-1)^{(3-p)}\Psi[i_v\phi]
\quadand
(L_v\Psi)[\phi]=-\Psi[L_v\phi]
\label{Defs_opps_on_Psi}
\end{gathered}
\end{align}
Thus for $\J$ to be closed requires 
\begin{align}
\J[d\lambda]=0
\label{Defs_J_closed}
\end{align}
for all test forms
$\lambda\in\Gamma_0\Lambda^0M$.

The monopole current $\Jmp$ is
defined in terms of the worldline $C:\DomC\to M$ where
$\DomC\subset\Real$ is the domain of the parameter $\tau$,
\begin{align}
\Jmp[\phi]=q\int_\DomC C^\star(\phi)
\label{Defs_Jmp}
\end{align}
where $C^\star:\Gamma_0\Lambda^1 M\to\Gamma_0\Lambda^1\DomC$ is the pullback.
Conservation of charge $d\Jmp=0$ implies $q$ is constant. In a
coordinate system this becomes (\ref{Intro_Jmp_a_delta}) and 
(\ref{Coord_Jmp_act_delta}).

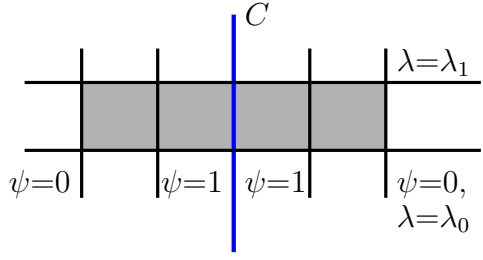
\begin{figure}
\centering
\begin{tikzpicture}[yscale=0.9]
\fill [black!30!white] (-2,1) rectangle +(4,1) ;
\draw [ultra thick,blue] (0,-0.5) -- +(0,3.5) ;
\draw  (0,3) node[right] {$C$} ;
\draw [very thick] (-1,0.3) -- +(0,2.2) ;
\draw [very thick] (-2,0.3) -- +(0,2.2) ;
\draw [very thick] (1,0.3) -- +(0,2.2) ;
\draw [very thick] (2,0.3) -- +(0,2.2) ;
\draw [very thick] (-2.75,1) -- +(6,0) ;
\draw [very thick] (-2.75,2) -- +(6,0) ;
\draw (-2,0.5) node[left] {$\psi{=}0$} 
(0,0.5) node[left] {$\psi{=}1$} 
(0,0.5) node[right] {$\psi{=}1$} 
(2,0.5) node[right] {$\psi{=}0,$} 
; 
\draw (2,2.3) node[right] {$\lambda{=}\lambda_1$} 
 (2,0.0) node[right] {$\lambda{=}\lambda_0$} ; 
\end{tikzpicture}
\caption{Construction of the test form
  $\psi\,d\lambda\in\Gamma_0\Lambda^1 M$. Its
  support is the shaded region.}
\label{fig_supp_mu_d_lambda}
\end{figure}

Higher order multipoles may be constructed by acting on $\Jmp$ with the
operations given in (\ref{Defs_opps_on_Psi}), and then ensuring that
the resulting distribution in closed.
Unlike the dipole/quadrupole relations where they mix, the monopole
can be separated off. That is all multipoles may be written
\begin{align}
\J_{\text{total}}=\Jmp+\J_{\textup{Monopole Free}}
\label{Defs_split_J}
\end{align}
for some value of the charge $q$. We say that $\J$ is \defn{monopole
  free} if for any 
\begin{align}
\J[\psi\,d\lambda]=0
\label{Defs_mono_free}
\end{align}
for all scalar fields $\lambda,\psi$ such that $\psi$ is flat in a
neighbourhood of $C$, $C^\star(\psi)=1$ and the combination
$\psi\,d\lambda$ has compact support.  See figure
\ref{fig_supp_mu_d_lambda}.  It is trivial to see that
\begin{align}
\Jmp[\psi\,d\lambda]
=
q\int_\DomC d\, C^\star(\lambda)
=
q (\lambda_1-\lambda_0)
\label{Defs_mono_q}
\end{align}
where $\displaystyle\lambda_1=\lim_{\tau\to\sup(\DomC)} \lambda(\tau)$ and 
$\displaystyle\lambda_0=\lim_{\tau\to\inf(\DomC)} \lambda(\tau)$.
We show in the appendix that $\J[\psi\,d\lambda]$ is independent of the
choice of $\lambda,\psi$ 
and hence (\ref{Defs_mono_q}) can be used to evaluate the charge
associated with a multipole.

The \defn{order of a multipole} is defined as follows. If
\begin{equation}
\begin{aligned}
\J[\lambda^{k+1} \phi]=0
\quadtext{for all}
&\lambda\in\Gamma\Lambda^0M\text{ and }
\phi\in\Gamma_0\Lambda^1M
\\&\quadtext{such that}
C^\star(\lambda)=0
\end{aligned}
\label{Defs_order}
\end{equation}
then we say that the order of $\J$ is at most $k$. Since we impose that
$\lambda$ vanishes on the image of $C$, this implies that we need to
differentiate the argument $\lambda^{k+1} \phi$ at least $k+1$ times
for $\J[\lambda^{k+1} \phi]\ne0$. We say dipoles have order at most one
and quadrupoles have order at most two. Therefore the terms in a
dipole have at most one derivative, and those in a quadrupole at most
two. This is consistent with the fact that the set of quadrupoles include all
dipoles.

As stated the electric multipoles can be defined in a metric-free and
coordinate-free manner, which contrast with the magnetic multipoles.
We say that $\J$ is an \defn{electric} multipole of order at most $\ell$  if
\begin{equation}
\begin{aligned}
\J[\lambda^{\ell} d\mu]=0
\quadtext{for all}&
\lambda,\mu\in\Gamma\Lambda^0M\quadtext{such that}
\\&
C^\star(\lambda)=C^\star(\mu)=0
\end{aligned}
\label{Defs_Elec_order}
\end{equation}
Clearly if $\J$ satisfies (\ref{Defs_order}) at order $k$ then it
satisfies (\ref{Defs_Elec_order}) at order $\ell=k+1$. Hence all dipoles are
electric quadrupoles. In the appendix we show that if $\J$ satisfies
(\ref{Defs_Elec_order}) at order $\ell$ then it also satisfies
(\ref{Defs_order}) at order $k=\ell$. Thus all electric quadrupoles are
quadrupoles.

For a dipole at rest, not satisfying (\ref{Defs_Elec_order}) with $\ell=1$,
i.e. $\Jdp[\lambda d\mu]\ne0$ for some $\lambda,\mu$ with
$C^\star(\lambda)=C^\star(\mu)=0$ then $\Jdp$ contains
magnetic dipole components. Likewise if a quadrupole $\Jqp$ does not
satisfy (\ref{Defs_Elec_order}) with $\ell=2$, we say it has magnetic components.

Using (\ref{Defs_mono_q}),(\ref{Defs_order}),(\ref{Defs_Elec_order}) we
can now define the multipoles we are interested in:
\begin{jgitemize}
\item
A monopole, $\Jmp$ is a zero order 3-form distribution over $C$.

\item
A dipole, $\Jdp$, is a closed, monopole free, 
first order 3-form distribution over $C$. This is equivalent to both
(\ref{Intro_Jdp_a_delta}) and (\ref{Coord_Jdp_act_delta}).
\item
An electric dipole, $\Jed$, is a dipole satisfying
  (\ref{Defs_Elec_order}) with $\ell=1$. This is equivalent to
(\ref{Intro_Jdp_a_delta}),(\ref{Coord_Jdp_act_delta}) together
with (\ref{Coord_DP_ED}).
\item
A quadrupole, $\Jqp$, is a closed, monopole free, 
second order 3-form distribution over $C$. This is equivalent to both
(\ref{Intro_Jqp_a_delta}) and (\ref{Coord_Jqp_act_delta}).
\item
An electric quadrupole, $\Jeq$, is a quadrupole satisfying
  (\ref{Defs_Elec_order}) with $\ell=2$. This is equivalent to
(\ref{Intro_Jqp_a_delta}),(\ref{Coord_Jqp_act_delta}) together
with (\ref{Coord_EQP}).
\end{jgitemize}
These equivalences are all demonstrated in the appendix.

\section{Conclusion and Discussion}
\label{ch_Conc}

In this article we have calculated the coordinate transformations
associated with quadrupoles and the their unusual
property, namely second order derivative and
integration. There is always a tension as to the pro and cons of the
using coordinate-free approaches. However given the complicated
coordinate transformation given here, it is the opinions of the
authors that the coordinate-free definition of quadrupoles is clearly
justified. We have shown that electric multipoles are more
``fundamental'' than magnetic multipoles since they can be defined
without a metric or preferred coordinate system.

This work raises many interesting questions and
directions one may pursue:
\begin{itemize}
\item
As stated, using the metric one may define a pure magnetic dipole. 
I.e. a dipole with no electric dipole
terms. However it is unknown to what extend one can define a
magnetic quadrupole which does not contain any electric terms.  Also
unknown is whether one can define an electric quadrupole which does
not contain magnetic dipole terms.
By contrast if one prescribes a laboratory coordinate system then one
can define all the objects: electric dipole (dim=3), magnetic dipole (dim=3), dipole free electric quadrupole (dim=6) and dipole free magnetic
quadrupole (dim=8). As stated these will mix with respect to other
coordinate systems. 
It is natural to extend this analysis to higher order multipoles. Raab and Lange
\cite{raab2005transformed} list the 77 electric terms up to octopole.

\item
It should be possible to extend this analysis to look at quadrupole
sources for linearised gravity. This is important as a source for
gravitational waves. In contrast to the closed $3-$form for
electromagnetic currents, the quadrupole in linearised gravity is a
stress-energy-momentum tensor. 

\item
As mentioned above, the results presented here can be extended not only
to higher dimensions but also to one and two dimensional sources,
i.e. which trace out world-sheets and three dimensional timelike
manifolds. One can even construct an event multipole, which has
support in just one event in spacetime.
One application of multipoles on higher dimensional manifolds is in
accelerator physics where the high energy bunch of electrons can be
expressed as a multipole expansion in seven dimension (phase space
$+$ time).

\item
There is a longstanding debate in the literature about the correct
equation of motion for a point charge that includes the back reaction,
with most authors favouring the Abraham-Lorentz-Dirac equation
\cite{poisson1999introduction,rohrlich2007classical,ferris2011origin}. This despite its well
documented pathologies. The problem for dipoles is more challenging as
one would have to renormalise a force which goes as $\sim r^{-3}$ as
one approaches the dipole.  Should that challenge be achieved and the
radiation reaction for quadrupoles be desired then the prescription
given here for the equations of motion will be needed.  Alternatively
a higher order theory of electromagnetism could be considered such as
the Bopp-Podolski theory \cite{gratus2015self}. In this theory the
distributional sources for the moving multipoles would still be valid.  
Hence the electromagnetic fields due to a dipole would, we conjecture, grow as $\sim r^{-2}$ as
$r\to0$ and hence may be renormalisable. Consequently the quadrupole
fields would grow as $\sim r^{-3}$.

\item
The fact that the transformation rules for quadrupoles involve an
integral poses the question about the bundle structure of
quadrupoles. Since $\gamma^{abc}$ is a function of $\tau$ one would
look for a vector bundle over $\DomC$, whose sections are in one to one
correspondence with the set of quadrupoles. In future work
\cite{gratus2017DeRham} we show
that such a vector bundle exists but is not unique and depends on a
choice of {\it thickening}, that is a domain $U\subset M$ and a map
$\Pi:U\to\DomC$ such that the combination $\Pi(C(\tau))=\tau$.

\end{itemize}

\appendix

\section{Statement and proof of the results in section \ref{ch_CoFree}}

\begin{lemma}
The classification of dipoles and quadrupoles: 

All dipoles $\Jdp$, i.e.
closed, monopole free,  first order $3-$form distribution
over $C$ are given by \eqref{Coord_Jdp_act_delta}.

All quadrupoles $\Jqp$, i.e.
closed, monopole free,  second order $3-$form distribution
over $C$ are given by \eqref{Coord_Jqp_act_delta}.
\end{lemma}

\begin{proof}

Let $\Jdp$ be given by (\ref{Coord_Jdp_act_delta}) and
$\lambda\in\Gamma\Lambda^0M$ with $C^\star(\lambda)=0$.  Since there
is only one derivative of $\lambda^2\phi_a$ then
$\Jdp[\lambda^2\phi]=0$ hence $\Jdp$ has degree at most one. Likewise
$\Jqp$ given by (\ref{Coord_Jqp_act_delta}) has degree at most two.
By requiring (\ref{Defs_J_closed}), i.e. 
$\Jdp[d\lambda]=0$ and $\Jqp[d\lambda]=0$ implies the symmetry
conditions on $\gamma^{ab}$ and $\gamma^{abc}$ respectively.

In order to show that a general dipole or quadrupole distribution with 
order $k$ according to
\eqref{Defs_order} can be written as (\ref{Coord_Jdp_act_delta}) or
(\ref{Coord_Jqp_act_delta})  we need to consider an adapted coordinate
system. Let $(z^0,z^1,z^2,z^3)$ be a coordinate system adapted to the
embedding $C$, so that $C(\tau)=(\tau,0,0,0)$.

By the manipulations in (\ref{Defs_opps_on_Psi}) we can see that the
general multipole can be constructed from just addition, internal
contraction, Lie derivatives and wedge products. Thus 
\begin{align*}
\Psi[\phi]
=&
\sum_{\textup{terms like}}\int_\DomC C^\star(L_{v_1}\cdots L_{v_r} \phi)
+
\sum_{\textup{terms like}} \int_\DomC C^\star(L_{v_1}\cdots L_{v_r} i_w\phi)\, d\tau
\end{align*}
For the dipole case then (\ref{Defs_order}) implies 
$\Jdp[(z^\mu)^2\phi]=0$ and
$\Jdp[(z^\mu+z^\nu)^2\phi]=0$ for $\mu=1,2,3$ and hence
$\Jdp[z^\mu z^\nu\phi]=0$. Thus there cannot be any terms with
$\displaystyle\pqfrac{\phi_a}{z^\mu}{z^\nu}\rule{0em}{2em}$. Thus the general degree
three distribution of order at most one is given by
\begin{align*}
\Jdp[\phi]
&=
\int_\DomC\Big(
{\Jcurr^{\emptyset,0}}\,\phi_0 
+
\sum_\mu {\Jcurr^{\mu,0}}\,\partial_\mu \phi_0
+
\sum_\nu {\Jcurr^{\emptyset,\nu}\,\phi_\nu}
+
\sum_{\mu,\nu} {\Jcurr^{\mu,\nu}} \partial_\mu\phi_\nu
\Big)d\tau
\end{align*}
where $\displaystyle\partial_\mu={\pfrac{}{z^\mu}}$. 
In this section (indexed) scalar fields (such as $\partial_c\phi_0$) are
implicitly evaluated on $C(\tau)$.
Set $d\Jdp=0$ and hence $\Jdp[d\lambda]=0$
for all $\lambda\in\Gamma_0\Lambda^0M$. This gives the following equations
\begin{align}
& 
{\dot\Jcurr^{\emptyset,0}}=0
\,,\qquad
{\dot\Jcurr^{\mu,0}} -
{\Jcurr^{\emptyset,\mu}} = 0
\qquadand
{\Jcurr^{\mu,\nu}} - {\Jcurr^{\nu,\mu}} = 0
\label{Coords_DP_cond}
\end{align}
where 
${\dot\Jcurr^{\mu,0}}=\displaystyle\dfrac{\Jcurr^{\mu,0}}{\tau}$.
The $\Jcurr^{\emptyset,0}=q$ gives the monopole
term. The remaining terms are then given by
\begin{equation}
\begin{aligned}
\Jdp[\phi] 
&= 
\int_\DomC\bigg(
\sum_{\mu} \Big( {\Jcurr^{\mu,0}}
\,\partial_\mu \phi_0
-
\dot\Jcurr^{\mu,0}\,\phi_\mu
\Big)
+
\sum_{\mu<\nu}
{\Jcurr^{\mu,\nu}}\, \big( \partial_\mu \phi_\nu - \partial_\nu \phi_\mu\big)
\bigg)\,d\tau
\\&= 
\int_\DomC\bigg(
\sum_{\mu} {\Jcurr^{\mu,0}}
\Big( \partial_\mu \phi_0
-
\partial_0\phi_\mu
\Big)
+
\sum_{\mu<\nu}
{\Jcurr^{\mu,\nu}}\, \big( \partial_\mu \phi_\nu - \partial_\nu \phi_\mu\big)
\bigg)\,d\tau
\end{aligned}
\label{Coords_DP_Res}
\end{equation}
Thus there are six free parameters. These correspond to the
$\gamma^{ab}$ for the adapted coordinates via 
$\gamma^{0\mu}=-{\Jcurr^{\mu,0}}$ and 
$\gamma^{\nu\mu}=-{\Jcurr^{\mu,\nu}}$. Since there are six free
parameters which is the same number in (\ref{Coord_Jdp_act_delta}), then
(\ref{Coord_Jdp_act_delta}) covers all the dipoles.


For quadrupoles the general degree
three distribution of order at most two is given by
\begin{align*}
\Jqp[\phi]
=
\int_\DomC\bigg(
{\Jcurr^{\emptyset,0}}\,\phi_0
+
\sum_{\mu} {\Jcurr^{\mu,0}}\,\partial_\mu \phi_0
+
\sum_{\nu} {\Jcurr^{\emptyset,\nu}\,\phi_\nu}
+
\sum_{\mu,\nu} {\Jcurr^{\mu,\nu}}\, \partial_\mu\phi_\nu
&+
\sum_{\mu\le\nu} 
\Jcurr^{\mu\nu,0}\,\partial_\mu \partial_\nu \phi_0
\\&
+
\sum_{\mu\le\nu}\sum_{\rho} \Jcurr^{\mu\nu,\rho}\,
\partial_\mu \partial_\nu\phi_\rho
\bigg) \,d\tau
\end{align*}
Again setting $d\Jqp=0$ implies
\begin{align*}
\begin{gathered}
{\Jcurr^{\emptyset,0}}=q
\,,\quad
{\dot\Jcurr^{\mu,0}} - 
{\Jcurr^{\emptyset,\mu}} = 0
\,,\quad
{\Jcurr^{\mu,\mu}}
-  
{\dot\Jcurr^{\mu\mu,0}}
= 0
\,,\quad
{\Jcurr^{\mu,\nu}}
+  {\Jcurr^{\nu,\mu}}
-  
{\dot\Jcurr^{\mu\nu,0}}
= 0\,,
\\
{\Jcurr^{\mu\mu,\mu}}=0
\,,\quad
{\Jcurr^{\mu\mu,\rho}}+
{\Jcurr^{\mu\rho,\mu}}=0
\quadand
{\Jcurr^{12,3}}+
{\Jcurr^{13,2}}+
{\Jcurr^{23,1}}=0
\end{gathered}
\end{align*}
for $\mu<\nu$ and $\mu\ne\rho$.
This gives
\begin{equation}
\begin{aligned}
\Jqp[\phi]
&=
\int_\DomC\bigg(
\sum_{\mu} {\Jcurr^{\mu,0}}
\Big( \partial_\mu \phi_0
-
\partial_0\phi_\mu
\Big)
+
\tfrac12\sum_{\mu<\nu}
\big(\Jcurr^{\mu,\nu}-\Jcurr^{\nu,\mu}\big)
\, \big( \partial_\mu \phi_\nu - \partial_\nu \phi_\mu\big)
\\&\qquad
+
\sum_{\mu\le\nu} 
\Jcurr^{\mu\nu,0}\,\big(\partial_\mu \partial_\nu \phi_0
-
\tfrac12
\partial_0\partial_\nu \phi_\mu -
\tfrac12\partial_0\partial_\mu \phi_\nu\big)
+
\sum_{\mu\ne\nu}
{\Jcurr^{\mu\mu,\nu}}\,
(\partial_\mu \partial_\mu \phi_\nu-\partial_\nu \partial_\mu \phi_\mu)
\\&\qquad
+
{\Jcurr^{12,3}}\,
(\partial_1 \partial_2 \phi_3-\partial_2 \partial_3 \phi_1)
+{\Jcurr^{13,2}}\,
(\partial_1 \partial_3 \phi_2-\partial_2 \partial_3 \phi_1)
\bigg)
d\tau
\end{aligned}
\label{QP_adapt_coords}
\end{equation}
Again the twenty, non monopole, 
parameters $\Jcurr^{ab,c}$ correspond to the twenty
free parameters $\gamma^{abc}$. These are given by
\begin{align}
\begin{gathered}
\gamma^{\mu\nu\rho} = {\Jcurr^{\nu\rho,\mu}}
\,,\quad
\gamma^{\nu\mu\mu} = {\Jcurr^{\mu\mu,\nu}}
\,,\quad
\gamma^{0\mu\mu} = {\Jcurr^{\mu\mu,0}}
\,,
\quad
\\\gamma^{0\mu\nu} = {\Jcurr^{\mu\nu,0}}
\,,\quad
\dot\gamma^{\mu0\nu} = {\Jcurr^{\nu,\mu}}
\,,\quad
\dot\gamma^{\mu00} =
{\Jcurr^{\mu,0}}
\end{gathered}
\label{SMD_Eg_QP_convert_J_gam}
\end{align}
for $\mu\ne\nu=1,2,3$. Recall (\ref{Intro_QP_gamma_sym}) implies
$\gamma^{abb}=-2\gamma^{bab}=-2\gamma^{bba}$. Observe that some of the
components $\gamma^{abc}$ are differentiated.
\end{proof}

\begin{lemma}
The monopole term given by \eqref{Defs_mono_q} is
      independent of $\lambda,\mu$.
\end{lemma}

\begin{proof} Substituting $\phi=\psi\,d\lambda$
  into (\ref{Coords_DP_Res}) and (\ref{QP_adapt_coords}).  
We can use for example
\begin{align*}
\int_\DomC&
\sum_{\mu<\nu}
{\Jcurr^{\mu,\nu}}\, \big( \partial_\mu (\psi\partial_\nu\lambda) 
- \partial_\nu (\psi\partial_\mu\lambda) \big)
\,d\tau
=
\int_\DomC
\sum_{\mu<\nu}
{\Jcurr^{\mu,\nu}}\, \big( \partial_\mu \partial_\nu\lambda
- \partial_\nu \partial_\mu\lambda \big)
\,d\tau
=0
\end{align*}
Similarly with the other terms. Thus the only non zero term is the
monopole term.
\end{proof}

\begin{lemma}
For electric dipoles the following definitions are equivalent
\begin{jgitemize}
\item Equation \eqref{Defs_Elec_order} with $\ell=1$.
\item Equation \eqref{Coord_DP_ED} in arbitrary coordinates.
\item In adapted coordinates with $\gamma^{\mu\nu}=0$ for all
  $\mu,\nu=1,2,3$. 
\end{jgitemize}
For electric quadrupoles the following definitions are equivalent
\begin{jgitemize}
\item Equation \eqref{Defs_Elec_order} with $\ell=2$.
\item Equation \eqref{Coord_EQP} in arbitrary coordinates.
\item In adapted coordinates with $\gamma^{\mu\nu\rho}=0$ for all
  $\mu,\nu=1,2,3$. 
\end{jgitemize}
\end{lemma}

\begin{proof}
To show that (\ref{Coord_DP_ED}) implies
(\ref{Defs_Elec_order}) with $\ell=1$ we have.
\begin{align*}
\Jed[\lambda\,d\mu]
&=
\int_\DomC \big(w^a \Cdot^b - w^b \Cdot^a\big)\,
\pfrac{}{x^a}\Big(\lambda\,\pfrac{\mu}{x^b}\Big)\,d\tau
\\&=
\int_\DomC \big(w^a \Cdot^b - w^b \Cdot^a\big)\,
\pfrac{\lambda}{x^a}\pfrac{\mu}{x^b}\,d\tau
\\&=
\int_\DomC 
\Big(w^a \dfrac{C^\star(\mu)}{\tau}\pfrac{\lambda}{x^a} 
- w^b \dfrac{C^\star(\lambda)}{\tau}\pfrac{\mu}{x^b}\Big)
=0
\end{align*}
since $C^\star(\mu)=C^\star(\lambda)=0$.

To show that (\ref{Defs_Elec_order}) with $\ell=1$ implies
(\ref{Coord_DP_ED}) then write $\Jdp$ in adapted
coordinates (\ref{Coords_DP_Res}). By acting on $(\psi\, z^\rho\,d
z^\sigma)$ where $\psi$ has compact support, we
have from (\ref{Coords_DP_Res})
\begin{align*}
\Jdp[\psi\,z^\rho\,d z^\sigma]
&=
\int_\DomC\Big(
\sum_{\mu} {\Jcurr^{\mu,0}}
\big( \partial_\mu (\psi\,z^\rho \delta^\sigma_0)
-
\partial_0 (\psi\,z^\rho \delta^\sigma_\mu)
\big)
+
\sum_{\mu<\nu}
{\Jcurr^{\mu,\nu}}\, 
\big( \partial_\mu (\psi\,z^\rho \delta^\sigma_\nu) - 
\partial_\nu (\psi\,z^\rho \delta^\sigma_\mu)\big)
\Big)\,d\tau
\\&=
\int_\DomC
{2\Jcurr^{\rho,\sigma}}\, 
C^\star(\psi)\,d\tau
\end{align*}
Since this is true for all $\psi$ we have
$\Jcurr^{\rho,\sigma}=0$. Thus we are left with three parameters
  $\Jcurr^{\mu,0}$. Thus both
  (\ref{Coord_DP_ED}) and (\ref{Defs_Elec_order}) are defined
  by three parameters and so they are equal.

Since $\gamma^{\mu\nu}=\Jcurr^{\nu,\mu}$ then the definitions in terms
of adapted coordinates follows.

\vspace{1em}

The proof for quadrupoles is the same as the proof of dipoles. First substitute (\ref{Coord_EQP}) into
(\ref{Defs_Elec_order}) with $\ell=2$ to show that
$\Jeq[\lambda^2\,d\mu]=0$. Then impose $\Jqp[\lambda^2\,d\mu]=0$ using
$\Jqp$ in adapted coordinates (\ref{QP_adapt_coords}).
This imply the eight terms $\Jcurr_0^{\mu\mu,\nu}=\Jcurr_0^{12,3}=\Jcurr_0^{13,2}=0$
for $\mu\ne\nu=1,2,3$. The remaining twelve terms give rise to the
electric quadrupoles.
\end{proof}

\begin{lemma}
If $\J$ satisfies \eqref{Defs_Elec_order} at order $\ell$ then $\J$
also satisfies \eqref{Defs_order} at order $k=\ell$.
\end{lemma}
\begin{proof}
Let $\J$ satisfies \eqref{Defs_Elec_order} as order $\ell$. By expanding
out powers of the form
$\J[(\lambda_{i_1}+\lambda_{i_2}+\ldots)^{\ell}d\mu]=0$ we can show
$\J[\lambda_1\lambda_2\cdots\lambda_\ell\,d\mu]=0$ for all
$\lambda_1,\ldots,\lambda_\ell,\mu\in\Gamma_0\Lambda^0M$ such that
$C^\star(\lambda_1)=\cdots=C^\star(\lambda_\ell)=C^\star(\mu)=0$. 

In adapted coordinates $\phi=\phi_0\, d\tau+\phi_\mu\, dz^\mu$. So 
\begin{align*}
\J[\lambda^{\ell+1}\phi]
&=
\J[\lambda^{\ell+1}(\phi_0\, d\tau+\phi_\mu\,  dz^\mu)]
\\&=
\J[\lambda^{\ell}\phi_0\,  d(\lambda\tau)]
-
\J[\lambda^{\ell}\phi_0\tau\,  d\lambda]+
\J[\lambda^{\ell+1}\phi_\mu \, dz^\mu]
\\&=
\J[\lambda^{\ell-1}(\lambda\phi_0)\,  d(\lambda\tau)]
-
\J[\lambda^{\ell-1}(\lambda\phi_0\tau)\,  d\lambda]+
\J[\lambda^{\ell}(\lambda\phi_\mu) \, dz^\mu]
=0
\end{align*}
\end{proof}


\section*{Acknowledgement}

{The authors are grateful for the support provided by STFC (the
Cockcroft Institute ST/G008248/1 and ST/P002056/1), EPSRC (the Alpha-X project
EP/J018171/1  and EP/N028694/1) 
and the Lancaster University Faculty of Science and
Technology studentship program.}

{The authors would like to thank Dr. David Burton and Dr. Paul Kinsler 
both of Lancaster University.
}

\bibliographystyle{unsrt}
\bibliography{Quad}

\end{document}